\documentclass[10pt, conference, letterpaper]{IEEEtran}
\IEEEoverridecommandlockouts

\usepackage{amssymb}
\usepackage{amsfonts}
\usepackage{amsmath}
\usepackage{amsthm}
\usepackage{float}
\usepackage{epsfig}
\usepackage{graphicx}
\usepackage{subcaption}
\usepackage{tabularx}

\usepackage{url}
\usepackage{amssymb}
\usepackage{amsfonts}
\usepackage{amsmath}
 
\usepackage{amsmath,amssymb,verbatim}
\usepackage{graphicx}
\usepackage{subcaption}
\usepackage{float}
\usepackage{epsfig}
\usepackage{threeparttable}
\usepackage{algorithm}
\usepackage{algorithmic}
\usepackage{color} 
\usepackage{dblfloatfix}
\usepackage{changepage}
\usepackage{comment}
\usepackage{multirow}
\usepackage{colortbl}
\usepackage{ulem}
\usepackage{color}
\usepackage{lipsum}
\usepackage{cite}

\usepackage{tabularx}

\floatname{algorithm}{Algorithm}

\newtheorem{theorem}{\it Theorem}

\newtheorem{lemma}{\it Lemma}

\begin{document}

\title{A Worst-Case Approximate Analysis of Peak Age-of-Information Via  Robust Queueing Approach}

\author{Zhongdong Liu, Yu Sang, Bin Li, and Bo Ji
\thanks{This work was supported in part by the NSF under Grants CCF-1657162, CNS-1651947, CNS-1717108, and CNS-1815563.}
\thanks{Zhongdong Liu (zhongdong@vt.edu) and Bo Ji (boji@vt.edu) are with the Department of Computer Science, Virginia Tech, Blacksburg, VA. 
Yu Sang (yu.sang@temple.edu) is with the Department of Computer and Information Sciences, Temple University, Philadelphia, PA. 
Bin Li (binli@uri.edu) is with the Department of Electrical, Computer and Biomedical Engineering, University of Rhode Island, Kingston, Rhode Island. 
}}

\maketitle







\begin{abstract}
A new timeliness metric, called \emph{Age-of-Information (AoI)}, has recently attracted a lot of research interests for real-time applications with information updates.
It has been extensively studied for various queueing models based on the probabilistic approaches, \textcolor{black}{where the analyses heavily depend on the properties of specific distributions (e.g., the memoryless property of the exponential distribution or the \textit{i.i.d.} assumption).}
In this work, we take an alternative new approach, the robust queueing approach, to analyze the \emph{Peak Age-of-Information (PAoI)}.
Specifically, we first model the uncertainty in the stochastic arrival and service processes using uncertainty sets.
This enables us to approximate the expected PAoI performance for very general arrival and service processes, including those exhibiting heavy-tailed behaviors \textcolor{black}{or correlations, where traditional probabilistic approaches cannot be applied.}
We then derive a new bound on the PAoI in the single-source single-server setting.
Furthermore, we generalize our analysis to two-source single-server systems with symmetric arrivals, which involves new challenges (e.g., the service times of the updates from two sources are coupled in one single uncertainty set).
Finally, through numerical experiments, we show that our new bounds provide a good approximation for the expected PAoI. 
Compared to some well-known bounds in the literature (e.g., one based on Kingman's bound \textcolor{black}{under the \textit{i.i.d.} assumption}) that tends to be inaccurate under light load, our new approximation is accurate under both light and high loads, both of which are critical scenarios for the AoI performance.
\end{abstract}

\section{Introduction}
The  last decades have witnessed significant advances in the computing capabilities of mobile and sensing devices.
The communication capacity of networks has also increased by orders of magnitude.
These developments have spawned a wide variety of real-time applications that require timely information updates. 
A practically important example of such emerging applications is autonomous vehicular systems, 
where real-time vehicular status information (location, velocity, acceleration, etc.) needs to be shared 
with nearby vehicles to enable full self-driving capability~\cite{vehicle1,vehicle2}.
Other examples include sensor networks for environmental monitoring, weather or news update applications, and live streaming services.

For such real-time services that require timely information updates, a major concern is about the freshness of the data delivered to the receiver.
Those commonly used metrics, such as throughput and delay, cannot precisely measure this timeliness related feature \cite{kaul2012real}.
To that end, in the seminal work~\cite{vehicle1}, a timeliness metric called \emph{Age-of-Information (AoI)} is proposed to measure the freshness of the received data.
It is defined as the time elapsed since the most recently received update was generated (see Eq.~\eqref{eq:age} for a formal definition).
In this work, we focus on the metric of \textit{Peak Age-of-Information (PAoI)} \cite{costa2014age}, which is defined as the maximum value of the AoI before it drops due to a newly delivered fresh update. 
The PAoI is a critical metric for certain time-sensitive applications that have very stringent timeliness requirements.
Consider autonomous vehicular systems, where
the status information must be shared with nearby vehicles.
The PAoI has to be strictly lower than a certain threshold at all times so that the autonomous vehicular controller makes correct decisions and thus ensures safety and efficiency.
Clearly, in such applications, it is insufficient to guarantee a certain level of average AoI only.

Most prior work takes a probabilistic approach to analyze and optimize the AoI and PAoI, which is based on the assumption that the interarrival time and service time follow certain distributions.
For example, the exponential distribution has played a privileged role in modeling stochastic systems \cite{kaul2012real,kaul2012status,costa2014age,sun2016update, sang2017power,pappas2015age,yates2012real,huang2015optimizing,dogan2020multi,moltafet2020approximate}.
In order to make the analysis tractable, the service times are often assumed to be identically and independently distributed (\textit{i.i.d.}).
Although these assumptions lead to tractable performance analysis and optimization, such assumptions may not hold in many practical scenarios. On the other hand, general distributions introduce significant challenges to near-exact analysis of the system performance.
It becomes even more challenging if the interarrival times and service times have heavy-tailed distributions and are potentially interdependent.

%

To that end, we take an alternative approach to model queueing systems based on \emph{robust optimization}~\cite{robustoptimization} and \emph{robust queueing theory}~\cite{robust}, which is originally developed for approximating the system time. 
Using this new analytical framework, we model the uncertainty in the stochastic arrival and service processes using \emph{uncertainty sets} and approach the problem using a robust optimization formulation. Note that the robust optimization theory can be used to model both light-tailed and heavy-tailed systems.
Therefore, our analysis no longer relies on the assumption of specific distributions with attractive properties (e.g., the memoryless property or even \textit{i.i.d.} arrival and service processes).
Instead, only the first and second order statistical information (i.e., mean interarrival/service time and variance) is required for the analysis.
For the single-source single-server system, we derive an upper bound on the worst-case system time under the assumption of uncertainty sets, which will be used to approximate the expected PAoI performance.
\textit{While Kingman's bound \cite{kingman1970inequalities,ciucu2018two} is also a well-known approximation of the system time, it requires \textit{i.i.d.} interarrival and/or service times. Also, both Kingman's bound and the original robust queueing analysis in \cite{robust} are accurate only when the traffic load is high. In contrast, our approach does not have such limitations.}
Furthermore, we generalize the analysis to the two-source single-server setting with symmetric traffic arrivals.
In this scenario, the updates from two different sources will be processed by a single server in a shared manner, which makes the PAoI analysis more challenging.

We summarize our main contributions as follows. 
\begin{itemize}
\item \emph{To the best of our knowledge, this is the first work that applies the robust-queueing approach to analyzing the PAoI performance in information-update systems.}
This new analytical framework can be applied to a wide range of queueing models without any assumption of specific distributions.
In particular, it works well for the systems with non-\textit{i.i.d.} interarrival  times  and  service  times.
\item We consider a single-source system and derive an upper bound on the PAoI. The upper bound can be used to develop approximations that are very close to the expected PAoI under both light and high traffic loads.
This is particularly important to the AoI analysis, as both long service times (when the load is high) and long interarrival times (when the load is light) would result in a large AoI.
\item We further generalize the analysis to the two-source setting with symmetric arrivals. The generalization is non-trivial and involves new challenges.
One key challenge is that the service times of the updates from two sources are coupled in one single uncertainty set.
Therefore, the property of uncertainty set cannot be directly used for analyzing the PAoI performance of each source.
\item Finally, we perform extensive numerical experiments and evaluate the PAoI performance
under different traffic loads as well as for different stochastic processes.
The simulation results show that our new bound with properly chosen parameters of uncertainty sets provides accurate approximations for the PAoI performance.
\end{itemize}

The remainder of this paper is organized as follows.
We discuss the related work on AoI and robust queueing theory in Section~\ref{sec:relatedwork}.
Then, we describe our model and provide our analysis for a single-source system in Section~\ref{sec:single}.
In Section~\ref{sec:multiple}, we generalize our results to the two-source case.
Finally, we present the numerical results in Section~\ref{sec:simulation}
and make concluding remarks in Section~\ref{sec:conclusion}.

\section{Related Work}\label{sec:relatedwork}
The AoI, a recently proposed metric, has inspired a series of studies on the analysis and optimization of the timeliness performance (see \cite{kosta2017age,sun2019age,2020arXiv200708564Y} for a survey).
The notion of AoI is formally introduced in \cite{kaul2012real}, where the authors analyze the time-average AoI in M/M/$1$, M/D/$1$, and D/M/$1$ systems under the First-Come-First-Served (FCFS) policy.
In \cite{kaul2012status}, the average AoI is analyzed for the M/M/$1$ system under the Last-Come-First-Served (LCFS) policy with and without preemption.  
In \cite{costa2014age}, the AoI performance of the FCFS policy in the M/M/$1$/$1$ and M/M/$1$/$2$ queues is studied, where new arrivals are discarded if the buffer is full. 
More sophisticated models have also been considered in the literature, such as two-source systems \cite{pappas2015age,yates2012real,huang2015optimizing} and multi-server systems \cite{bedewy2016optimizing,sang2017power,2019arXiv191208722L}.
However, most of the previous studies adopt the traditional probabilistic analytical framework and assume that the interarrival time and service time follow certain distributions.


In order to overcome the limitations of the probabilistic framework, recent work also considers other approaches to calculate the average AoI. 
For example, the authors in \cite{inoue2019general} derive the stationary distribution of the AoI, which is in terms of the stationary distribution of the delay and the PAoI. With the AoI distribution, one can analyze the mean or higher moments of the AoI in GI/GI/1, M/GI/1, and GI/M/1 queues under several scheduling policies (e.g., FCFS and LCFS). 
The authors in \cite{champati2019distribution} characterize the violation probability of AoI and use it to obtain upper bounds of AoI for GI/GI/$1$/$1$ and GI/GI/$1$/$2^*$ systems.
However, these studies are still based on the assumption of \textit{i.i.d.} random variables, which is not required in our approach. 
In \cite{yates2018age}, the stochastic hybrid system (SHS) is introduced as an analytical technique for studying the AoI. The SHS method provides a way to derive closed-form AoI results for simple queues described as finite-state Markov chains.
However, the SHS-based technique is limited to finite-state systems only (e.g., finite-buffer queueing systems) and is inapplicable to our problem.

In this paper, we propose an alternative approach based on the robust queueing framework to analyze the AoI.
In what follows, we briefly discuss recent developments of robust optimization and robust queueing theory.
Robust optimization has been proven to be an efficient approach for complex optimization problems with significant uncertainties \cite{bandi2012tractable, robustoptimization}.
This approach is later adopted in the development of robust queueing theory \cite{robust}, 
which can be used to provide fairly accurate predictions of the performance of complex queueing systems without making any probabilistic assumptions.
Note that in the traditional probabilistic approaches, it is often assumed that the underlying queueing system has the Markovian property (i.e., exponentially distributed interarrival times and service times) and 
that the interarrival times and the service times are both \textit{i.i.d.}
In contrast, using the robust-queueing approach, we no longer make such probabilistic assumptions.
Motivated by the Central Limit Theorem (CLT) and its generalized version, the randomness is modeled as uncertainty sets rather than specific probabilistic distributions.
Moreover, the interarrival times and service times do not have to be \textit{i.i.d.} over time \cite{robust}.
This enables us to characterize the queueing performance for very general arrival and service processes, including those exhibiting heavy-tailed behaviors or correlations.
In this paper, we will adopt this framework for the PAoI analysis in information-update systems.

\section{Single-Source System}\label{sec:single}
In this section, we discuss a simple case with one source, one server and one monitor.
We then generalize the analysis to the setting with two pairs of source and monitor in Section~\ref{sec:multiple}.

\subsection{System Model}\label{sec:singlemodel}

In a single-source information-update system, there is only one pair of source and monitor as illustrated in Fig.~\ref{fig:single_model}.
Each update is stamped with the time when it is generated.
Let $a_n$ denote the generation time of the $n$-th update.
Then the interarrival time between the $n$-th update and the $(n-1)$-st update can be denoted by $T_n \triangleq a_n - a_{n-1}$.
We assume that the interarrival time follows a general distribution with mean $1/\lambda$. 
After an update is generated, it needs to be processed by the server before it is delivered to the monitor.
The server has a FCFS queue of infinite buffer size.
We assume that the update arrives to the queue immediately after being generated. 
Hence, the $n$-th update arrives to the queue at time $a_n$ as well.
Let $W_n$ denote the waiting time of the $n$-th update.
The service time of the $n$-th update is denoted by $X_n$, which also follows a general distribution with mean $1/\mu$.
Let $f_n$ denote the time when the service of the $n$-th update at the server is finished.
After the update completes its service, it will be immediately delivered to the monitor.
Therefore, the update arrives at the monitor also at time $f_n$.
Let $S_n$ denote the total system time experienced by the $n$-th update, which is also equal to the sum of its waiting time in the queue and the service time,
\begin{equation}
    S_n \triangleq  f_n - a_n = W_n + X_n.
\end{equation}

\begin{figure}[!t]
\centering
\includegraphics[width=0.35\textwidth]{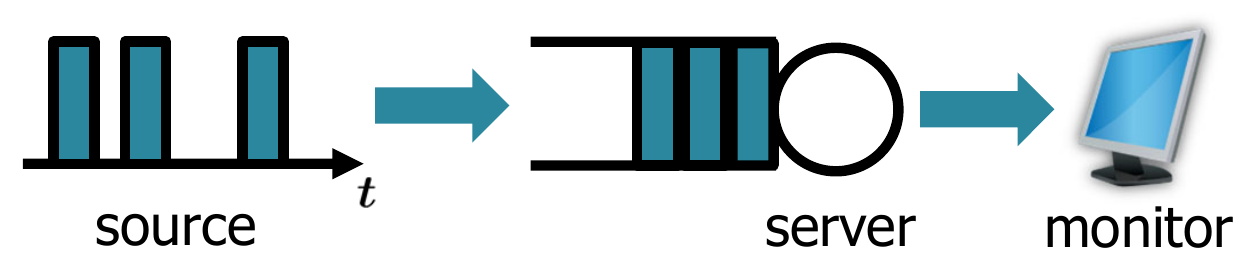}
\caption{A single-source information-update system.}
\label{fig:single_model}
\end{figure}

For such a time-sensitive information-update system, we are interested in the freshness of data at the monitor.
The freshness of data is measured by the metric \emph{Age-of-Information} (AoI).
It is defined as the time elapsed since the freshest update received by the monitor was generated. 
Assume that the latest update received by the monitor at time $t$ is stamped with the generation time $a(t)$, then the AoI  at the monitor is
\begin{equation}\label{eq:age}
\Delta(t) \triangleq  t - a(t).
\end{equation}
An example of the evolution of the AoI at the monitor is shown in Fig.~\ref{fig:sawtooth}.
The AoI would increase linearly when no update is transmitted to the monitor, and it reaches a local maximum value immediately before a new update is delivered.
Such a maximum value is called the \emph{Peak Age-of-Information (PAoI)}. There is a PAoI corresponding to each update in our example.
\textcolor{black}{Let $P_n$ denote the $n$-th PAoI. From Fig.~\ref{fig:sawtooth}, it is easy to see ${P_n} = {f_n} - {a_{n-1}}$. This can be rewritten as the sum of the interarrival time between the $n$-th and $(n-1)$-st updates (i.e., ${T_n}={a_n} - {a_{n - 1}}$) and the system time of the $n$-th update (i.e., ${S_n}={f_n} - {a_n}$). Then, the expected PAoI can be expressed as}
\begin{equation}
    \label{eq:e_paoi}
    \mathbb{E}\left[ P_n \right] =\mathbb{E}\left[ T_n \right] + \mathbb{E}\left[ S_n \right].
\end{equation}

It is important to study the PAoI performance since it represents the locally largest values of the AoI at the monitor and captures more stringent timeliness requirements.
Most of the previous work takes a probabilistic approach to analyze the PAoI, which assumes specific types of distribution of the service and interarrival times.
Taking a different path, we adopt a worst-case approximation approach based on the robust queueing theory, which only requires the knowledge of the first and second order statistical information of the distributions.

\begin{figure}[!t]
\centering
\includegraphics[width=0.38\textwidth]{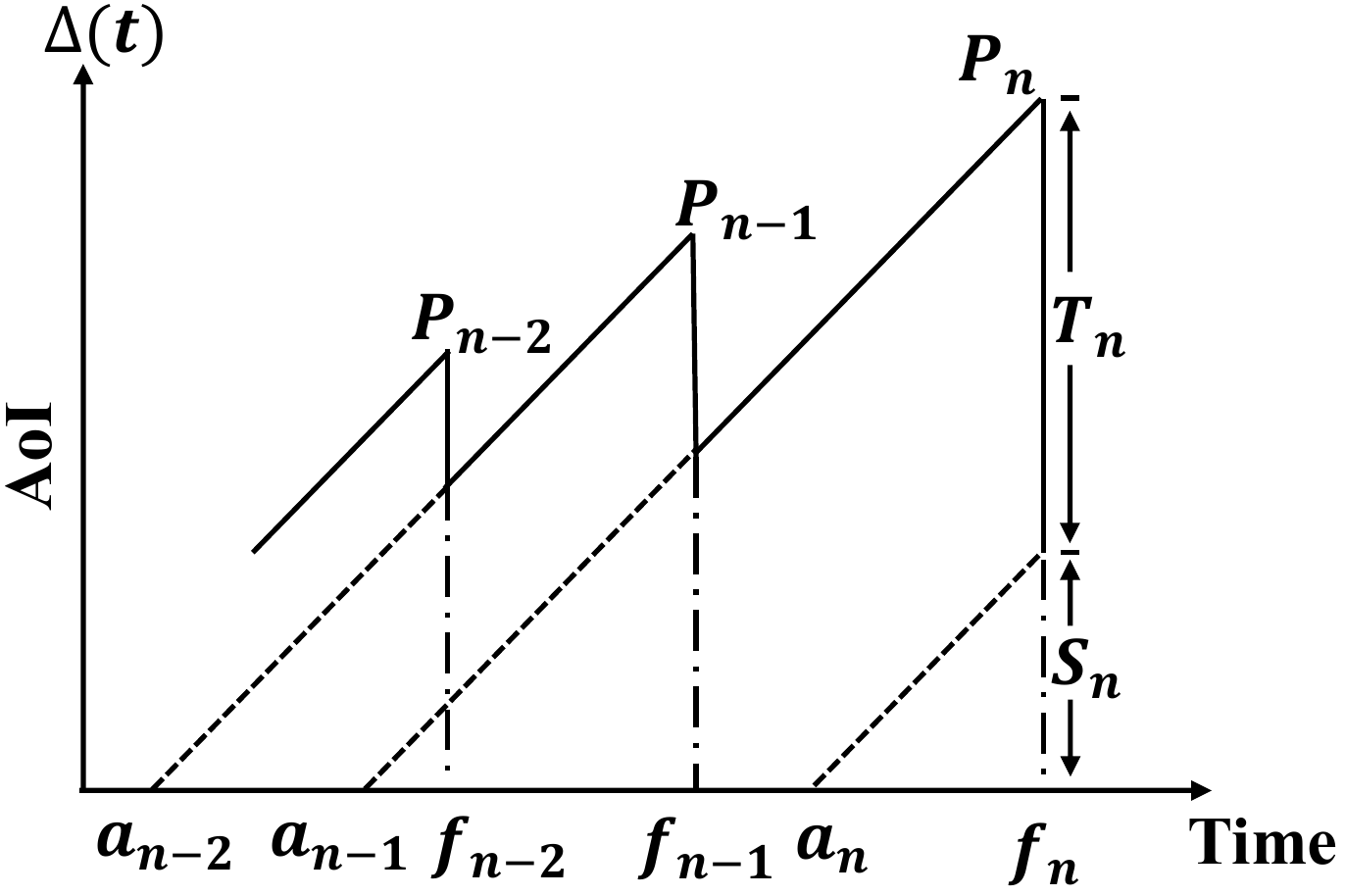}
\caption{An illustration of the AoI evolution at the monitor in the single-source system.}
\label{fig:sawtooth}
\end{figure}

\subsection{The Worst-Case Approach} 
\label{subsec:worst-case_single}

\textcolor{black}{
In this subsection, we will analyze the expected PAoI (see Eq.~\eqref{eq:e_paoi}). Since we trivially have $\mathbb{E}\left[ T_n \right] = 1/\lambda$, it remains to analyze the expected system time $\mathbb{E}\left[ S_n \right]$. 
}

Similar to \cite{robust}, we consider a sample path of $n$ updates. The generation and service process of the updates can be characterized by the interarrival times $\mathbf{T}_n \triangleq (T_1, T_2, \dots, T_n)$ and the service times  $\mathbf{X}_n \triangleq (X_1, X_2, \dots, X_n)$, respectively. Applying the robust queueing theory proposed in \cite{robust}, we adopt a worst-case analysis to obtain an upper bound on the system time $S_n$, which will be leveraged to approximate the expected system $\mathbb{E}\left[ S_n \right]$.
\textit{Different from \cite{robust} that neglects the accuracy of ${S_n}$ under light load, we characterize ${S_n}$ under both light and high loads in the single-source system and the two-source system} (see more discussions in Section~\ref{sec:single_result} and Section~\ref{sec:multiple}, respectively). 
To begin with, we express the system time of the $n$-th update $S_n$ as follows \cite{systemtime}: 
\begin{align} \label{eq:singlepeak}
S_n 
& = W_n + X_n \notag \\
& = \max_{1\leq k\leq n} \left( \sum_{i = k}^n X_i - \sum_{i = k+1}^n T_i \right),
\end{align}
where we recall that $W_n$ ($X_n$, resp.) is the waiting time (the service time, resp.) of the $n$-th update.
Note that $S_n$ involves interarrival times $\mathbf{T}_n$ and service times $\mathbf{X}_n$.


From Eq.~\eqref{eq:singlepeak}, we can see that the analysis of the system time involves the sum of system time and interarrival time of previous updates, which often makes the probabilistic approach intractable.
On the other hand, the Generalized Central Limit Theorem (GCLT) indicates that the distribution of the sum of multiple \textit{i.i.d.} random variables converges to a stable distribution. Motivated by GCLT, we assume that the partial sum of service times satisfies the following:
\begin{equation}\label{eq:singlexineq}
\frac{\sum_{i = k}^n X_i - (n-k+1)/\mu}{(n-k+1)^{1/\alpha}} \leq \Gamma_s,
\end{equation}
where ${\rm{1/}}\mu $ is the expected service time, $\alpha\in(1,2]$ is the tail coefficient that models possibly heavy-tailed probability distributions\footnote{The heavy-tailed distributions (like the Pareto or the Weibull) have heavier tails than the exponential distribution. Roughly speaking, there is a larger probability of getting very large values.} (the closer to $1$, the heavier the tail), and $\Gamma_s>0$ is the variability parameter chosen to ensure that Eq.~\eqref{eq:singlexineq} is satisfied with high probability. 
\textcolor{black}{For instance, the normalized sum of a large number of positive Pareto random variables can be approximated by a random variable $Y$ following a standard stable distribution with a tail coefficient $\alpha $ and ${C_\alpha } = {[G (1 - \alpha )\cos (\pi \alpha /2)]^{1/\alpha }}$, where $G ( \cdot )$ represents the gamma function. For a tail coefficient of $\alpha  = 1.5$, we obtain $\mathbb{P}\left( {Y \le 6.5} \right) \approx 0.975$ and $\mathbb{P}\left( {Y \le 19} \right) \approx 0.995$ \cite{robust}. Thus, choosing ${\Gamma _s} = 6.5$  is one good option in this instance}. 
\textcolor{black}{Note that there is a tradeoff between robustness and accuracy when choosing ${\Gamma _s}$. That is, when choosing a larger enough ${\Gamma _s}$, we can ensure that Eq.~\eqref{eq:singlexineq} is satisfied with a higher probability (i.e., stronger robustness), but this may result in a loose bound on the system time (i.e., lower accuracy). On the other hand, when choosing a smaller ${\Gamma _s}$, we can obtain tight bounds on the system time, but Eq.~\eqref{eq:singlexineq} may be violated with a higher probability. In Section~\ref{subsec:vp}, we discuss one way of choosing ${\Gamma _s}$ that attempts to balance the tradeoff.}
With a properly chosen ${\Gamma _s}$, we further assume that the service times $\mathbf{X}_n$ belong to the following parameterized uncertainty set \cite{robust}:
\begin{align} \label{eq:service-Uset}
\mathcal{U}_s \triangleq \Bigg\{ \mathbf{X}_n ~\vert~ & \frac{\sum_{i = k}^n X_i -  (n-k+1)/\mu}{(n-k+1)^{1/\alpha}}  \leq \Gamma_s, \notag \\
& \forall ~1\leq k \leq n \Bigg\}.
\end{align} 
Although the uncertainty set is motivated by the \textit{i.i.d.} assumption, $\mathbf{X}_n \in \mathcal{U}_s$ does not necessarily require that ${X_1},{X_2}, \ldots ,{X_n}$ be \textit{i.i.d.}

Similarly, we assume that the interarrival times $\mathbf{T}_n$ belong to the following uncertainty set:
\begin{align} \label{eq:interarrival-Uset}
\mathcal{U}_a \triangleq  \Bigg\{ \mathbf{T}_n ~\vert~ &  \frac{\sum_{i = k+1}^n T_i - (n-k)/\lambda}{(n-k)^{1/\alpha} } \geq -\Gamma_a, \notag \\
& \forall ~ 0\leq k\leq n-1  \Bigg\},
\end{align}
where ${\rm{1/}}\lambda $ is the mean interarrival time, $\alpha\in(1,2]$ is the tail coefficient that models possibly heavy-tailed probability distributions, and $\Gamma_a>0$ is the variability parameter chosen to ensure that Eq.~\eqref{eq:interarrival-Uset}  is satisfied  with  high probability. We consider the lower bound on the partial sums of the interarrival times as that leads to the worst-case system time.
We assume that the distributions of the service time and interarrival time have the same tail coefficient $\alpha$. \textcolor{black}{However, we assume no dependence between $\mathbf{T}_n$ and $\mathbf{X}_n$ and will discuss the case with dependence towards the end of this paper.}

\subsection{Performance Analysis}\label{sec:single_result}

Let $\hat{S}_n$ denote the worst-case system time of the $n$-th update. From Eq.~\eqref{eq:singlepeak}, we have the following:
\begin{align}\label{eq:singleworsttx}
\hat{S}_n \triangleq & \max_{ \mathbf{T}_n \in \mathcal{U}_a} \max_{\mathbf{X}_n\in \mathcal{U}_s} \max_{1\leq k \leq n} \left(\sum_{i = k}^n X_i - \sum_{i = k+1}^n T_i \right) \notag \\ \le & \textcolor{black}{\mathop {\max }\limits_{1 \le k \le n} \left( {\mathop {\max }\limits_{{{\bf{X}}_n} \in {{\cal U}_s}} \sum\limits_{i = k}^n {{X_i}}  - \mathop {\min }\limits_{{{\bf{T}}_n} \in {{\cal U}_a}} \sum\limits_{i = k + 1}^n {{T_i}} } \right).}
\end{align}
It is shown in \cite{robust} that we can find the sample path of $\hat{\bf{X}}_n \in \mathcal{U}_s$ and $\hat{\bf{T}}_n \in \mathcal{U}_a$ that achieve the worst-case system time (\textcolor{black}{i.e., $\sum\nolimits_{i = k}^n {{{\hat X}_i}}  = \mathop {\max }\limits_{{{\bf{X}}_n} \in {{\cal U}_s}} \sum\nolimits_{i = k}^n {{X_i}} $ and $\sum\nolimits_{i = k+1}^n {{{\hat T}_i}}  = \mathop {\min }\limits_{{{\bf{T}}_n} \in {{\cal U}_a}} \sum\nolimits_{i = k+1}^n {{T_i}} $}).
Plugging such $\hat{\bf{T}}_n$ and $\hat{\bf{X}}_n$ into Eq.~\eqref{eq:singleworsttx} gives the exact expression of $\hat{S}_n$:
\begin{align} 
{\hat{S}_n} = \max_{1\leq k \leq n} \Bigg\{ \frac{n-k+1}{\mu} & - \frac{n-k}{\lambda} + \Gamma_s(n-k+1)^{1/\alpha} \notag \\ + & \Gamma_a(n-k)^{1/\alpha} \Bigg\}.
\label{eq:exacthsn}
\end{align}
As we can see from Eq.~\eqref{eq:exacthsn}, the worst-case system time $\hat{S}_n$ is proportional to the mean service time $1/\mu$ and inversely proportional to the mean interarrival time $1/\lambda$. Also, larger $\Gamma_s$ and $\Gamma_a$ lead to a larger $\hat{S}_n$, given that $\Gamma_s$ and $\Gamma_a$ are used to bound the sum of service times $\mathbf{X}_n$ and the sum of interarrvial times $\mathbf{T}_n$ in the uncertainty sets, respectively. 

Based on Eq.~\eqref{eq:exacthsn} and some relaxations, one can derive an upper bound on the worst-case system time ${\hat S_n}$.
We restate this result in the following lemma.

\begin{lemma}[Theorem 2 in \cite{robust}]
In a single-source FCFS queueing system with $\mathbf{T}_n\in\mathcal{U}_a$ and $\mathbf{X}_n\in\mathcal{U}_s$, we have 
\begin{equation}\label{eq:singlebound}
{\hat S_n} \leq \frac{\alpha-1}{\alpha^{\alpha/(\alpha - 1)}}  \cdot \frac{(\Gamma_s + \Gamma_a)^{\alpha/(\alpha - 1)}}{(1/\lambda - 1/\mu)^{1/(\alpha-1)}} + 1/\lambda.
\end{equation}
\end{lemma}


It has been shown that the above upper bound is nearly tight when the traffic load is high \cite{robust}.
However, the relaxations used in \cite{robust} renders the bound loose when the traffic load is light, i.e., when $\lambda$ is relatively small.
This is acceptable in the analysis of queueing system with respect to traditional metrics, such as delay or throughput, which are typically pronounced in the high-load regime.
\textit{
While for the metric of PAoI, a low arrival rate means long interarrival times between consecutive updates, which leads to a poor PAoI performance as well. This indicates that both light and high loads affect the PAoI performance significantly, and thus, focusing on the system time in the high-load regime only is insufficient.  
}
Instead of introducing relaxations, we propose an alternative method that provides the exact characterization of the worst-case system time.
We derive a new upper bound of the worst-case system time, which is nearly tight under both light and high loads. The upper bound is presented in the following theorem.

\begin{theorem}\label{theo:singlebound}
Define the following function 
$f(m) \triangleq (m + 1)/\mu  - m/\lambda  + {\Gamma _s}{(m + 1)^{1/\alpha }} + {\Gamma _a}{m^{1/\alpha }}$
and let 
\begin{equation}
    l = \Bigg(\frac{\alpha(1/\lambda-1/\mu)}{\Gamma_a+\Gamma_s} \Bigg)^{\alpha/(1-\alpha)},
\end{equation}
where $\alpha  \in (1,2]$.
In a single-server FCFS queueing system with $\mathbf{T}_n\in\mathcal{U}_a$ and $\mathbf{X}_n\in\mathcal{U}_s$, we have
\begin{equation}\label{eq:singlebound2}
{{\hat S}_n} \leq \left\{ {\begin{array}{*{20}{c}}
{\max \{ f(n - 1),0\} ,\;\;\;{\rm{if}}\;n - 1 \le \left\lfloor l \right\rfloor  - 1}\\
{\max \{ f({m^*}),0\} ,\;\;\;\;\;{\rm{otherwise}}\;\;\;\;\;\;\;\;\;}
\end{array}} \right. ,
\end{equation}
where  ${m^*} \in \mathop {\arg \max }\limits_{m \in \{ \lfloor l \rfloor  - 1, \lfloor l \rfloor, \lfloor l \rfloor  + 1 \} \cap [0, n - 1]} f\left( m \right)$.
\end{theorem}

\begin{proof}
The derivation of the upper bound of the worst-case system time ${\hat S}_n$ in Eq.~\eqref{eq:singlebound2} utilizes the convacity of Eq.~\eqref{eq:exacthsn}.
Let $m = n-k $.
According to Eq.~\eqref{eq:exacthsn}, we can rewrite the worst-case system time as
\begin{align}
   {\hat S_n}  & = \mathop {\max }\limits_{_{0 \le m \le n - 1}} \frac{{m + 1}}{\mu } - \frac{m}{\lambda } + {\Gamma _s}{(m + 1)^{1/\alpha}} + {\Gamma _a}{m^{1/\alpha }} \notag 
   \\&= \mathop {\max }\limits_{_{0 \le m \le n - 1}} f(m).
\end{align}    
Let $m^*$ be the integral maximizer of $f(m)$, i.e., ${\hat S_n} = f(m^*)$.
In order to find the bound of ${m^{\rm{*}}}$, we extend the domain of $f(m)$ to the set of nonnegative real numbers (i.e., $m \in R^{+}$). 
The second order derivative of $f(m)$ is 
\begin{equation}
f^{''}(m) = \frac{1-\alpha}{\alpha^2}\Big(  \Gamma_s(m+1)^{(1-2\alpha)/\alpha} + \Gamma_a m^{(1-2\alpha)/\alpha}  \Big),
\end{equation}
which is negative since $\alpha  \in \left( {1,2} \right]$. This implies that $f(m)$ is concave.
Let $M\in R^+$ be the continuous maximizer of $f(m)$, i.e., $M$ is the solution of the following equation:
\begin{equation} \label{eq:single-first-deri}
    f'(m) = \frac{1}{\alpha } \left({\Gamma _s}{(m + 1)^{\frac{{1 - \alpha }}{\alpha }}} + {\Gamma _a}{m^{\frac{{1 - \alpha }}{\alpha }}} \right) - \frac{1}{\lambda } + \frac{1}{\mu } = 0.
\end{equation}
However, it is usually difficult to solve Eq.~\eqref{eq:single-first-deri} to get the expression of $M$.
We then define function $g\left( m \right)$ as
\begin{equation}
    g\left( m \right) = \frac{1}{\alpha }({\Gamma _s}{m^{\frac{{1 - \alpha }}{\alpha }}} + {\Gamma _a}{m^{\frac{{1 - \alpha }}{\alpha }}}) - \frac{1}{\lambda } + \frac{1}{\mu },
\end{equation}
and let  $l$ be the solution of $g\left( m \right) = 0$, which gives
\begin{equation} \label{eq:g(l)}
    g\left( l \right) = \frac{1}{\alpha }({\Gamma _s}{l^{\frac{{1 - \alpha }}{\alpha }}} + {\Gamma _a}{l^{\frac{{1 - \alpha }}{\alpha }}}) - \frac{1}{\lambda } + \frac{1}{\mu } = 0,
\end{equation}
and 
\begin{equation} \label{eq:l}
    l = {\left( {\frac{{\alpha (1/\lambda  - 1/\mu )}}{{{\Gamma _a} + {\Gamma _s}}}} \right)^{\frac{\alpha }{{1 - \alpha }}}}.
\end{equation}
Note that the following is satisfied: 
\begin{align}
    f'\left( l \right) &= \frac{1}{\alpha }\left( {{\Gamma _s}{{\left( {l + 1} \right)}^{\frac{{1 - \alpha }}{\alpha }}} + {\Gamma _a}{{l}^{\frac{{1 - \alpha }}{\alpha }}}} \right) - \frac{1}{\lambda } + \frac{1}{\mu } \notag \\& = \frac{{{\Gamma _s}}}{\alpha }\left( {{{\left( {l + 1} \right)}^{\frac{{1 - \alpha }}{\alpha }}} - {{l}^{\frac{{1 - \alpha }}{\alpha }}}} \right)  \notag \\& < 0,
\end{align}  
where the second equality follows from Eq.~\eqref{eq:g(l)} and the last inequality holds because $x^{\frac{1 - \alpha}{\alpha }}$ is a decreasing function for $\alpha  \in \left( {1,2} \right]$. Similarly, we can show $f'\left( {l - 1} \right) > 0$.  Therefore,  we have $f'\left( {l - 1} \right) > f'\left( {M} \right) > f'\left( l \right)$, which implies $l-1 < M < l$ since $f'\left( m \right)$ is a decreasing function due to the concavity of $f(m)$.
Some thoughts give that the integral maximizer $m^*$ must satisfy the following: $m^*\in\{ \lfloor l\rfloor - 1, \lfloor l\rfloor, \lfloor l\rfloor +1\}$.
Recall that function $f\left( m \right)$ is defined on  $[0,n - 1]$.
If ${n - 1 \le \left\lfloor l \right\rfloor  - 1}$, i.e., the integral maximizer ${{m^*}}$ is out of the domain,  then we have ${{\hat S}_n} \le f\left( {n - 1} \right)$  since  $f\left( m \right)$ is an increasing function on $[0,n - 1]$;
otherwise, if $n - 1 > \left\lfloor l \right\rfloor  - 1$, then ${{m^*}}$ must be a value belonging to $\left\{ {\left\lfloor l \right\rfloor  - 1,\left\lfloor l \right\rfloor ,\left\lfloor l \right\rfloor  + 1} \right\}$ that satisfies ${{m^*} \le n - 1}$ and achieves the maximum of $f(m)$. 
Finally, we take the maximum between $f(m^*)$ and $0$ because the system time is nonnegative. This completes the proof.
\end{proof}    

\textcolor{black}{Using the upper bound of the worst-case system time derived in Theorem~\ref{theo:singlebound}, we can approximate the expected system time by choosing appropriate variability parameters $\Gamma_a$ and $\Gamma_s$ (see Section~\ref{subsec:vp} for detailed discussions).
This enables us to approximate the expected PAoI given that the expected interarrival time is already known (i.e., $1/\lambda$). 
}



\begin{figure}[!t]
\centering
\includegraphics[width=0.35\textwidth]{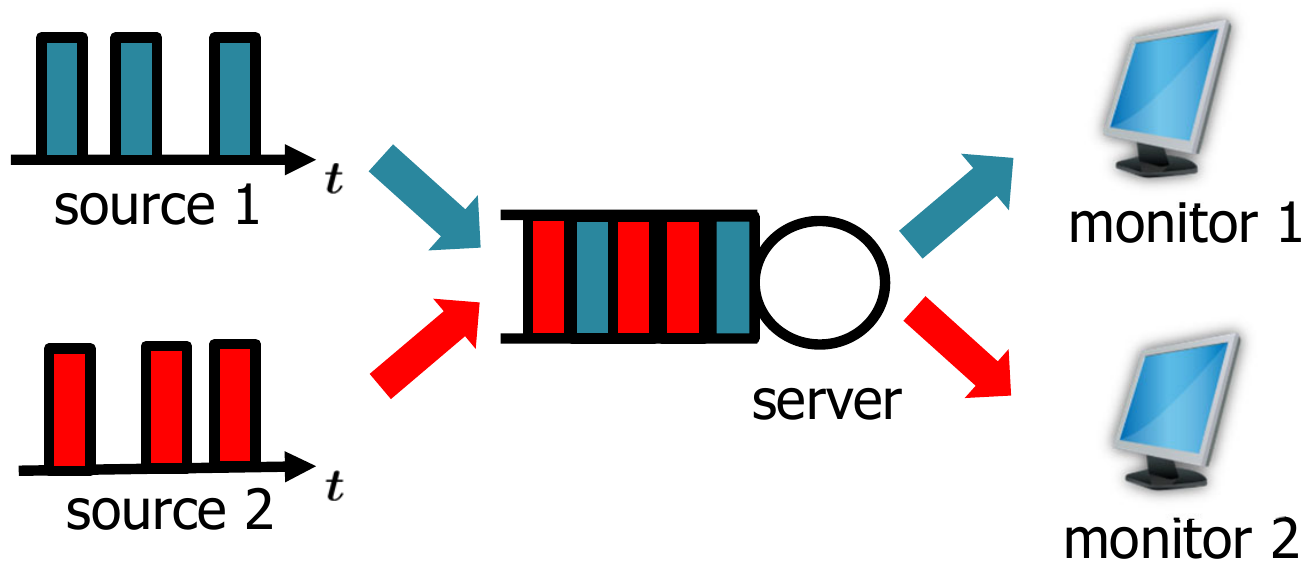}
\caption{A two-source information-update system.}
\label{fig:multiple_model}
\end{figure}

\section{Two-Source System}\label{sec:multiple}
In this section, we consider a more general system that consists of two sources, one shared server, and two monitors. 
Each monitor receives the update only from one source.
An illustration of the system is shown in Fig.~\ref{fig:multiple_model}.
Our interest is to analyze the expected PAoI performance at each monitor.
\textcolor{black}{
Note that this generalization is non-trivial and involves new technical challenges (e.g., the service times of the updates from two sources are coupled in one single uncertainty set). }
 
\subsection{System Model}
There are a total number of $n$ updates from both sources all going through one single server to the corresponding monitor.
We assume that these two sources are symmetric, i.e., their update interarrival time follows the same general distribution with mean $1/\lambda$ and their update's service time also follows the same general distribution with mean  $1/\mu$.
We assume  $2\lambda /\mu< ~1$ in order to keep the system stable.
We use $a^{(s)}_n$ ($f^{(s)}_n$, resp.) to denote the arrival time (service completion time, resp.) of the $n$-th update from source $s = 1, 2$.
The service time of the $n$-th update from source $s$ is denoted by $X^{(s)}_n  \triangleq f^{(s)}_n - a^{(s)}_{n}$. 
The interarrival time between the $n$-th update and the $(n - 1)$-st update from source $s$ is denoted by $T^{(s)}_n  \triangleq a^{(s)}_n - a^{(s)}_{n - 1}$. 
Then, the AoI at monitor $s$ is defined as 
\begin{equation}
\Delta^{(s)}(t) \triangleq t - a^{(s)}(t),
\end{equation}
where $a^{(s)}(t)$ denotes the generation time of the most recently received update at monitor $s$.
\textcolor{black}{Similarly, the $n$-th PAoI at monitor $s$ can be expressed as  $P_n^{(s)} \triangleq f_n^{(s)} - a_{n - 1}^{(s)} = T_n^{(s)} + S_n^{(s)}$, and the expected PAoI can be denoted as 
\begin{equation}
    \mathbb{E}[P_n^{(s)}] = \mathbb{E}[T_n^{(s)}] + \mathbb{E}[S_n^{(s)}] = 1/\lambda  + \mathbb{E}[S_n^{(s)}].
\end{equation}
Similar to the single-source case, we utilize the robust queueing theory to analyze the worst-case performance of $S_n$ (the system time of the $n$-th update, either from source 1 or 2) and use it to approximate $\mathbb{E}[S_n^{(s)}]$.}

Since the updates from both sources will join the same FCFS queue, we reorder all updates according to their arrival time. 
Let $a_n$ denote the arrival time of the $n$-th update that arrives at the server\footnote{
We use the same notations as those in the single-source case. 
If it has a superscript, the superscript indicates which source this notation is related
to. If not, this notation corresponds to the server side.}.
Note that this update could be from either source 1 or 2.
Consider a sample path of $n$ updates that arrive at the server at time $\mathbf{A}_n = \{a_1, a_2, \dots, a_n\}$ such that $a_1 \leq a_2 \leq \dots \leq a_n$.
The service time of the $n$-th update is denoted by $X_n$, and  corresponding service times for the whole sample path are $\mathbf{X_n} = \{X_1, X_2, \dots, X_n\}$.
The interarrival time between the $n$-th and $(n-1)$-st updates is denoted by $T_n \triangleq a_n - a_{n-1}$, and the corresponding interarrival times for the whole sample path are $\mathbf{T}_n = \{T_1, T_2, \dots, T_n\}$.
To distinguish the updates from source 1 and source 2, we further define the following functions that map the update to its source according to the arrival times:
\begin{align}\notag
b_s(n) \triangleq  \arg\min_k a_k^{(s)}\geq a_n,~ 
e_s(n) \triangleq  \arg\max_k &  a_k^{(s)} \leq a_n, \\& s = 1, 2. \notag
\end{align}
According to the mapping function, $e_s(n)$ is the index of the last update from source $s$ and $a_{e_s(n)}^{(s)}$ is the arrival time of this last update.
Therefore, 
the interarrival times of two sources can be denoted as ${\mathbf{T}_{{e_1}(n)}^{(1)}} = \{ T_1^{(1)}, T_2^{(1)}, \dots, T_{e_1(n)}^{(1)} \}$ and ${\mathbf{T}_{{e_2}(n)}^{(2)}} = \{ T_1^{(2)}, T_2^{(2)}, \dots, T_{e_2(n)}^{(2)} \}$, respectively.


Now, we consider the system time corresponding to the last update that arrives at the server  (i.e., the $n$-th update). 
Assume that this update is from source 1. Recall that the system time $S_n$ (see Eq.~\eqref{eq:singlepeak}) can be written as
\begin{align}
S_n = \max_{1\leq k \leq n} (\sum_{i = k}^n X_i - \sum_{i = k+1}^n T_i).\label{eq:mulsystem}
\end{align}
In the rest of the paper, the analysis is based on the assumption that  the $n$-th update is from source 1.
It is easy to apply the same analysis to the case that the $n$-th update is from source~2.

For the two-source system, the AoI is determined by service times $\mathbf{X}_n$ and the interarrival times of two sources, $\mathbf{T}_{{e_1}(n)}^{(1)}$ and $\mathbf{T}_{{e_2}(n)}^{(2)}$.
We need three uncertainty sets for them.
First, for the service times, we assume that two sources have the same type of updates.
Therefore, the same assumption as in the single-source system still holds since we also have only one server, i.e., 
\begin{align}
\mathcal{U}_s = \Bigg\{ (X_1, X_2, \dots, X_n) \vert & \frac{\sum_{i = k}^n X_i - \frac{n-k+1}{\mu}}{(n-k+1)^{1/\alpha}} \leq \Gamma_s, \notag \\
& \forall ~1\leq k \leq n \Bigg\}.
\end{align} 
For the interarrival time $\mathbf{T}_{{e_1}(n)}^{(1)}$ and $\mathbf{T}_{{e_2}(n)}^{(2)}$, we assume that the uncertainty set for source $s$ is
\begin{align}\label{eq:set_multi_interarrival}
\mathcal{U}_a^s =  \Bigg\{  (T^{(s)}_1, T^{(s)}_2, &  \dots,  T_{e_s(n)}^{(s)} ) \vert \frac{\sum_{i = k}^l T^{(s)}_i - \frac{l-k+1}{\lambda}}{(l-k+1)^{1/\alpha} } \geq -\Gamma_a,\notag\\
& \forall ~1\leq l \leq e_s(n) ~\text{and}~ 1\leq k\leq l  \Bigg\}.
\end{align}
Here the assumptions for $\alpha $, ${\Gamma _s}$, and ${\Gamma _a}$ are the same as that in the single-source case. With the uncertainty sets for each source, we further define the uncertainty set for the whole sample path, $\mathcal{U}_a=\mathcal{U}_a^1 \cup \mathcal{U}_a^2$.

\subsection{Performance Analysis}\label{sec:multi_peak}
In this subsection, we derive an upper bound on the worst-case system time in the two-source system.

\begin{theorem}\label{theo:multibound}
    Define the following function\footnote{By slightly abusing the notation, we also use $f(m)$ here.}
$f(m) \triangleq 2\left( {m + 1} \right)/\mu  - m/\lambda  + 2{\Gamma _s}{(m + 1)^{1/\alpha }} + {\Gamma _a}{m^{1/\alpha }}$, and let 
\begin{equation}
    l = \Big(\frac{\alpha(1/\lambda-2/\mu)}{\Gamma_a+2\Gamma_s} \Big)^{\frac{\alpha}{1-\alpha}},
\end{equation}
where $\alpha  \in (1,2]$.
In a single-server FCFS queue with two symmetric sources, such that $\mathbf{T}_{{e_s}(n)}^{(s)}\in\mathcal{U}_a^{s} \left( {s = 1,2} \right)$ and $\mathbf{X}_n\in\mathcal{U}_s$, we have 
\begin{equation}\label{eq:multibound}
{\hat S_n} \leq \left\{ {\begin{array}{*{20}{c}}
{\max \left\{ {f\left( { - \frac{1}{2}} \right),f\left( {\frac{n}{2} - 1} \right),0} \right\},\;\;\;{\rm{if}}\;\frac{n}{2} - 1 \le \left\lfloor l \right\rfloor  - 1}\\
{\max \left\{ {f\left( { - \frac{1}{2}} \right),f\left( {{m^*}} \right),0} \right\},\;\;\;\;\;\;\;\;\;\;\;\;{\rm{otherwise}}}
\end{array}} \right.,
\end{equation}
where ${m^*} \in \mathop {\arg \max }\limits_{m \in \mathcal{D} \cap \left\{ {m \le n/2 - 1} \right\}} f\left( m \right)$  and $\mathcal{D} \triangleq \{ \left\lfloor l \right\rfloor - 1, \left\lfloor l \right\rfloor - 1/2, \left\lfloor l \right\rfloor, \left\lfloor l \right\rfloor + 1/2, \left\lfloor l \right\rfloor  + 1\}$.
\end{theorem}


\begin{proof}
As shown in Eq.~\eqref{eq:multibound}, the expression of the system time is very similar to that in the single-source system.
However, the analysis is quite different since the interarrival times consist of two sequences.
Given $\mathbf{T}_{{e_1}(n)}^{(1)}$ and $\mathbf{T}_{{e_2}(n)}^{(2)}$, the worst-case system time is
\begin{align}\label{eq:worstt}
\hat{S}_n ({\mathbf{T}_{{e_1}(n)}^{(1)}}, {\mathbf{T}_{{e_2}(n)}^{(2)}}) 
 \leq
\max_{1\leq k \leq n} 
(\max_{\mathbf{X}_n\in \mathcal{U}_s}  
\sum_{i = k}^n X_i - \sum_{i = k+1}^n T_i).
\end{align}
First, it is shown in \cite{robust} that there exists a sequence of service times $ \mathbf{\hat{X}}_n \in \mathcal{U}_s$ that achieves the upper bound in Eq.~\eqref{eq:worstt}.
We restate the results here that the sequence $\mathbf{\hat{X}}_n$ satisfies: 
\begin{align}\notag
\sum_{i=k}^n \hat{X}_i = \max_{\mathbf{X}_n\in \mathcal{U}_s}  
\sum_{i = k}^n X_i =&  \frac{n-k+1}{\mu} +\Gamma_s (n-k+1)^{1/\alpha},\\
& \forall~ k = 1, 2, \dots, n.\label{eq:x}
\end{align}
This implies that the service times that achieve the worst-case system time are independent of the interarrival times.
Therefore, we can replace the partial sum of service times with Eq.~\eqref{eq:x}:
\begin{align}\label{eq:worst}
{\hat S_n} 
& = \max_{\mathbf{T}_n \in \mathcal{U}_a} \max_{1\leq k \leq n} (\sum_{i = k}^n \hat{X}_i - \sum_{i = k+1}^n T_i) \notag \\
& \leq
\max_{1\leq k \leq n} 
\bigg(\sum_{i = k}^n \hat{X}_i -  
\min_{\mathbf{T}_n \in \mathcal{U}_a}\sum_{i = k+1}^n T_i \bigg) 
\end{align}
Different from the service times, the interarrival times are from two different uncertainty sets and the number of updates from each source  depends on each other.
There do not exist two interarrival time sequences that achieve the maximum value of the system time for all possible $k$'s as we show in the single-source case.
As such, we define function $S(k)$ for all $k = 1, 2, \dots, n$ as 
\begin{equation}
S(k) \triangleq \sum_{i = k}^n \hat{X}_i - \min_{\mathbf{T}_n \in \mathcal{U}_a} \sum_{i = k+1}^n T_i,
\end{equation}  
and maximize the value of $S(k)$ for every $1\leq k\leq n$.
Therefore, the worst-case system time can be rewriten as
\begin{equation}\label{eq:worstbypk}
\hat{S}_n = \max_{1\leq k \leq n} S(k).
\end{equation}
Note that $S(k)$ is completely determined by the term of interarrival times' partial sum $T(k)$, which is defined as
\begin{equation}
T(k) \triangleq \min_{\mathbf{T}_n \in \mathcal{U}_a} \sum_{i = k+1}^n T_i.
\end{equation}
Then, deriving an upper bound of $S(k)$ is equivalent to deriving a lower bound of $T(k)$.
The partial sum of the interarrival times from source 1 and source 2 are both less than $T(k)$, i.e.,
\begin{align}\label{eq:two}
T(k) \geq \min _{\mathbf{T}_{{e_1}(n)}^{(1)} \in \mathcal{U}_{a}^{1}} \sum_{i=b_{1}(k)+1}^{e_{1}(n)} T_{i}^{(1)} \text { and }  \\  T(k) \geq \min _{\mathbf{T}_{{e_2}(n)}^{(2)} \in \mathcal{U}_{a}^{2}} \sum_{i=b_{2}(k)+1}^{e_{2}(n)} T_{i}^{(2)}.
\end{align} 
As shown in uncertainty sets $\mathcal{U}_{a}^{1}$ and $\mathcal{U}_{a}^{2}$, the lower bound of the partial sum of interarrival times depends on the number of updates.
The total number of updates from the $k$-th update to the $n$-th update is $(n - k + 1)$.
Let $h$ be the number of updates from source 2, i.e., $e_2(n) - b_2(k) + 1 = h$.
Correspondingly, we have $e_1(n) - b_1(k) + 1 = n -k + 1 - h$ updates from source $1$.
According to the uncertainty sets Eq.~\eqref{eq:set_multi_interarrival}, we have the following two lower bounds:
\begin{align}
\sum_{i=b_{1}(k)+1}^{e_{1}(n)} T_{i}^{(1)} & \geq \frac{e_{1}(n)-b_{1}(k)}{\lambda}-\Gamma_{a}\left(e_{1}(n)-b_{1}(k)\right)^{1 / \alpha} \notag \\
&=\frac{n-k-h}{\lambda}-\Gamma_{a}(n-k-h)^{1 / \alpha},
\label{eq:bound1m}\\
\sum_{i=b_{2}(k)+1}^{e_{2}(n)} T_{i}^{(2)} & \geq \frac{e_{2}(n)-b_{2}(k)}{\lambda}-\Gamma_{a}\left(e_{2}(n)-b_{2}(k)\right)^{1 / \alpha} \notag \\
&=\frac{h-1}{\lambda}-\Gamma_{a}(h-1)^{1 / \alpha}.
\label{eq:bound2m}
\end{align}
We denote the lower bounds in Eq.~\eqref{eq:bound1m} and Eq.~\eqref{eq:bound2m} as $L_1(h) =  (n - k - h)/{\lambda} - \Gamma_a(n - k - h)^{1/\alpha}$ and $L_2(h) = (h - 1)/{\lambda} - \Gamma_a(h-1)^{1/\alpha}$, respectively.
Therefore, we have 
\begin{align}
T(k)  
& \geq \min_{0 \leq h \leq n-k}  \max \{L_1(h), L_2(h)\}. \label{eq:minmax}
\end{align}
Note that both $L_1(h)$ and $L_2(h)$ are convex functions and they are symmetric along the line of $h = (n-k+1)/2$, which implies that the continuous minimizer of Eq.~\eqref{eq:minmax} is $h^{\prime} = (n-k+1)/2$.
Therefore, the lower bound for $T(k)$ is 
\begin{equation}\label{eq:tkbound}
T(k) \geq (n-k-1)/(2\lambda) - \Gamma_a((n-k-1)/2)^{1/\alpha}.
\end{equation}  
By replacing the partial sum $T(k)$ with its lower bound in Eq.~\eqref{eq:tkbound}, we define the upper bound of $S(k)$ as $\hat{S}(k)$, i.e., 
\begin{align} \label{eq:S_k}
\hat{S}(k) & = \sum_{i = k}^n \hat{X}_i -\frac{n-k-1}{2\lambda} + \Gamma_a(\frac{n-k-1}{2\lambda})^{1/\alpha} \notag\\
& = \frac{n-k+1}{\mu} + \Gamma_s (n-k+1)^{1/\alpha} \notag\\
& \quad  -\frac{n-k-1}{2\lambda} + \Gamma_a(\frac{n-k-1}{2})^{1/\alpha}.
\end{align}
Therefore, we have $S(k)\leq \hat{S}(k)$ for all $1\leq k\leq n$.
Then, we can rewrite the upper bound of the system time as 
\begin{equation} \label{eq:two-source-bound}
{\hat S_n} \leq \max_{1\leq k\leq n} \hat{S}(k).
\end{equation}
Similar to the proof of Theorem~\ref{theo:singlebound}, let $m = (n - k - 1){\rm{/2}} \in \left\{ { - {\rm{1/2}},0,{\rm{1/2}}, \ldots ,n/2 - 1} \right\}$.
Then, according to Eq.\eqref{eq:S_k} and the defined $f(m)$, 
Eq.\eqref{eq:two-source-bound} can be rewritten  as 
\begin{equation} \label{eq:two-source-sep-bound}
{{\hat S}_n} \le \mathop {\max }\limits_{ - \frac{1}{2} \le m \le \frac{n}{2} - 1} f(m) = \max \left\{ {f\left( { - \frac{1}{2}} \right),\mathop {\max }\limits_{0 \le m \le \frac{n}{2} - 1} f(m)} \right\}.
\end{equation}
Next, we consider the second item in max function of Eq.~\eqref{eq:two-source-sep-bound}. 
Let $m^*$ be the maximizer, i.e., ${f({m^*}) = \mathop {\max }\limits_{0 \le m \le \frac{n}{2} - 1} f(m)}$.
We can also extend $m \in \left\{ {0,1/2, \ldots ,n/2 - 1} \right\}$ to the real numbers,  i.e., $m\in R^+$.
It is easy to check that $f(m)$ is also concave since we have
\begin{equation}
f^{''}(m) = \frac{1-\alpha}{\alpha^2}\Big(   2\Gamma_s(m+1)^{\frac{1-2\alpha}{\alpha}} + \Gamma_a m^{\frac{1-2\alpha}{\alpha}}  \Big) < 0.
\end{equation}
Let $M \in R^+$ be the continuous maximizer of $f(m)$, i.e., 
\begin{equation}
f'\left( {M} \right) = \frac{2}{\mu } - \frac{1}{\lambda } + \frac{1}{2}\left( {2{\Gamma _s}{{(M + 1)}^{\frac{{1 - \alpha }}{\alpha }}} + {\Gamma _a}{{M}^{\frac{{1 - \alpha }}{\alpha }}}} \right) = 0.
\end{equation}
In order to obtain the range of $M$,  let $l$ be the solution of the following equation, 
\begin{equation}\label{eq:der2}
\frac{2}{\mu } - \frac{1}{\lambda } + \frac{1}{\alpha }\left( {2{\Gamma _s}{m^{\frac{{1 - \alpha }}{\alpha }}} + {\Gamma _a}{m^{\frac{{1 - \alpha }}{\alpha }}}} \right) = 0,
\end{equation}
which gives 
\begin{equation}
l = {\left( {\frac{{\alpha (1/\lambda  - 2/\mu )}}{{{\Gamma _a} + 2{\Gamma _s}}}} \right)^{\frac{\alpha }{{1 - \alpha }}}}.
\end{equation}
Similar to the single-source case, we have $l - 1 < M < l$.
Therefore, there must be ${m^*} \in \{ \left\lfloor l \right\rfloor  - 1,\left\lfloor l \right\rfloor  - 1/2,\left\lfloor l \right\rfloor ,\left\lfloor l \right\rfloor  + 1/2,\left\lfloor l \right\rfloor  + 1\} $. 
Recall that the concave function $f(m)$ is defined on $[0,n/2 - 1]$. 
If ${n/2 - 1 \le \left\lfloor l \right\rfloor  - 1}$, then we have ${{\hat S}_n} \le f\left( {n/2 - 1} \right)$; otherwise, ${{m^*}}$ must be a value among $\left\{ {\left\lfloor l \right\rfloor  - 1,\left\lfloor l \right\rfloor  - 1/2,\left\lfloor l \right\rfloor ,\left\lfloor l \right\rfloor  + 1/2,\left\lfloor l \right\rfloor  + 1} \right\}$ that satisfies ${{m^*} \le n/2 - 1}$ and  achieves the maximum of $f(m)$. 
Together with  Eq.~\eqref{eq:two-source-sep-bound} and the fact that the system time is nonnegative, we complete the proof.  
\end{proof}

\textcolor{black}{By choosing appropriate variability parameters ($\Gamma_a$ and $\Gamma_s$), we can accurately approximate the expected system time and thus the expected PAoI at each monitor, given that the expected interarrival time is already known (i.e., $1/\lambda$).}



\section{Numerical Results}\label{sec:simulation}

In this section, we perform extensive simulations to evaluate the accuracy of our theoretical results.
We first introduce how to approximate the expected system time of the steady-state queueing networks with the worst-case system times proposed in Theorem~\ref{theo:singlebound} or Theorem~\ref{theo:multibound}.
Then, for the single-source system, we show that our results in Theorem~\ref{theo:singlebound} can approximate the expected PAoI much better than the bounds in the literature (e.g., one bases on the Kingman's bound), especially in the light load case.
In the end, we also show that the bound in Theorem~\ref{theo:multibound} can also approximate the expected PAoI in two-source case very well.


\subsection{Variability Parameters}
\label{subsec:vp}
Note that the bounds in Theorem~\ref{theo:singlebound} and Theorem~\ref{theo:multibound} do not depend on the specific distribution of the interarrival time and service time.
The update arrival process and service process are fully characterized by the primitive data $(\lambda, \sigma_a^2)$ and $(\mu, \sigma_s^2)$, respectively, where $\sigma_a^2$ and $\sigma_s^2$ denote the variance of the interarrival time and service time, respectively. Therefore, it remains to translate the stochastic primitive data into uncertainty sets with appropriate variability parameters $(\Gamma_a, \Gamma_s)$ such that the bounds proposed in Theorem~\ref{theo:singlebound} and Theorem~\ref{theo:multibound} can approximate the expected system time of steady-state queueing networks well.

Inspired by the Kingman's bound\footnote{The Kingman's bound shows that the expected system time can be bounded by $\mathbb{E}\left[S\right] \leqslant \frac{\lambda}{2} \cdot \frac{\sigma_{a}^{2}+\sigma_{s}^{2}}{1-\rho}+\frac{1}{\mu}$.}\cite{kingman1970inequalities}, a mapping function is provided in \cite{robust} that describes the variability parameters in terms of the distributions' first and second statistics,
\begin{equation}
    \Gamma_a = \sigma_a, \Gamma_s = (\theta_0 + \theta_1\sigma_s^2 + \theta_2\sigma_a^2\rho^2)^{1/2}-\sigma_a,
\end{equation}
where $\rho = \lambda/\mu$ is the traffic density and $({\theta _0},{\theta _1},{\theta _2})$  are constants that can be derived from linear regression.
Specifically, in order to obtain appropriate values of $({\theta _0},{\theta _1},{\theta _2})$, we first simulate multiple instances of the queue for various parameters of $(\rho ,{\sigma _a},{\sigma _s})$ and different arrival and service distributions. Then we employ the linear regression to generate appropriate values for $({\theta _0},{\theta _1},{\theta _2})$ to adapt the value ${{\hat S}_n}$ obtained in Theorem~\ref{theo:singlebound} (or Theorem~\ref{theo:multibound}) to the expected value of the simulated system time. 
This allows us to build a dictionary or a look-up table of variability parameters values for given arrival and service distributions that makes the following approximation $\mathbb{E}[S(\mathbf{T},\mathbf{X})] \approx {{\hat S}_n}({\Gamma _a},{\Gamma _s})$. Table \ref{table:adaption} provides the resulting $({\theta _0},{\theta _1},{\theta _2})$ for each adaption regimes.

\begin{table}[!t]
\centering
\begin{tabular}{|c|c|c|}
\hline
$({\theta _0},{\theta _1},{\theta _2})$ & Single-Source & Two-Source \\ \hline
${\theta _0}$ & -0.376 & -1.302 \\ \hline
${\theta _1}$ & 3.978 & 6.021 \\ \hline
${\theta _2}$ & 0.5 & 0.7 \\ \hline
\end{tabular}
\caption{Service adaptation regimes.}
\label{table:adaption}
\end{table}

\begin{table}[!t]
\centering
\begin{tabular}{|c|c|c|c|c|}
\hline
\multicolumn{1}{|l|}{}                                                   & Methods                 & Exponential & Normal     & Uniform    \\ \hline
\multirow{3}{*}{\begin{tabular}[c]{@{}c@{}}Single-\\ Source\end{tabular}} & Kingman's bound        & $33.86\% $  & $22.58\% $ & $14.89\% $ \\ \cline{2-5}                        & Robust Approx. 1 & $32.01\% $  & $34.90\% $ & $36.49\% $ \\ \cline{2-5} 
& Robust Approx. 2 & $8.32\% $   & $8.47\% $  & $9.28\% $  \\ \hline
\begin{tabular}[c]{@{}c@{}}Two-\\ Source\end{tabular}   & Robust Approx. 3 & $12.68\% $  & $10.05\% $ & $9.79\% $  \\ \hline
\end{tabular}
\caption{Error percent of different approximation methods.}
\label{table:error_percent_single_source}
\end{table}

\begin{figure*}[t]
\centering
\begin{subfigure}[b]{0.329\linewidth}
	\centering
	\includegraphics[width=1\textwidth]{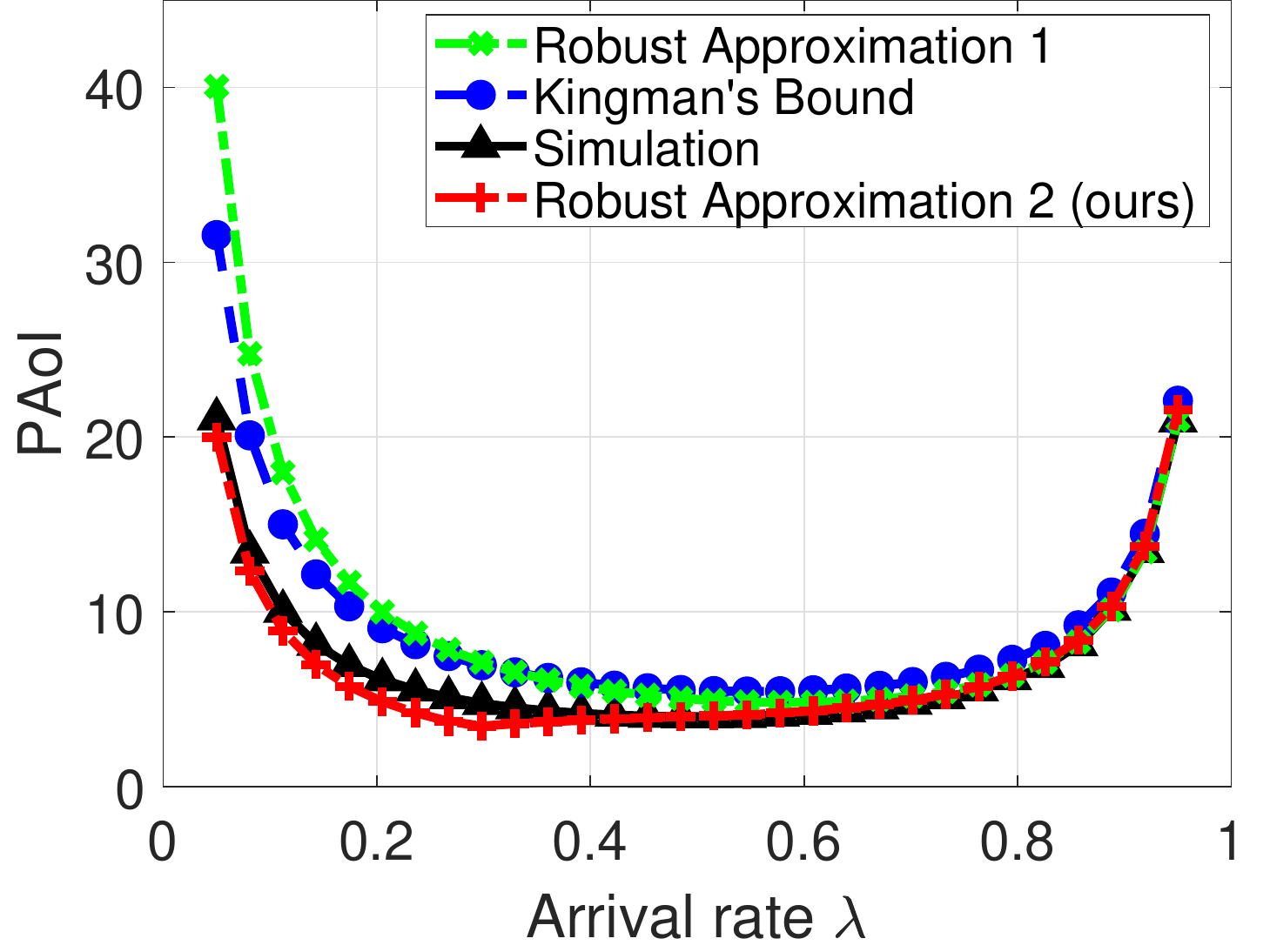}
	\caption{Exponential distribution}
	\label{fig:single_poisson}
\end{subfigure}
\begin{subfigure}[b]{0.329\linewidth}
	\centering
	\includegraphics[width=1\textwidth]{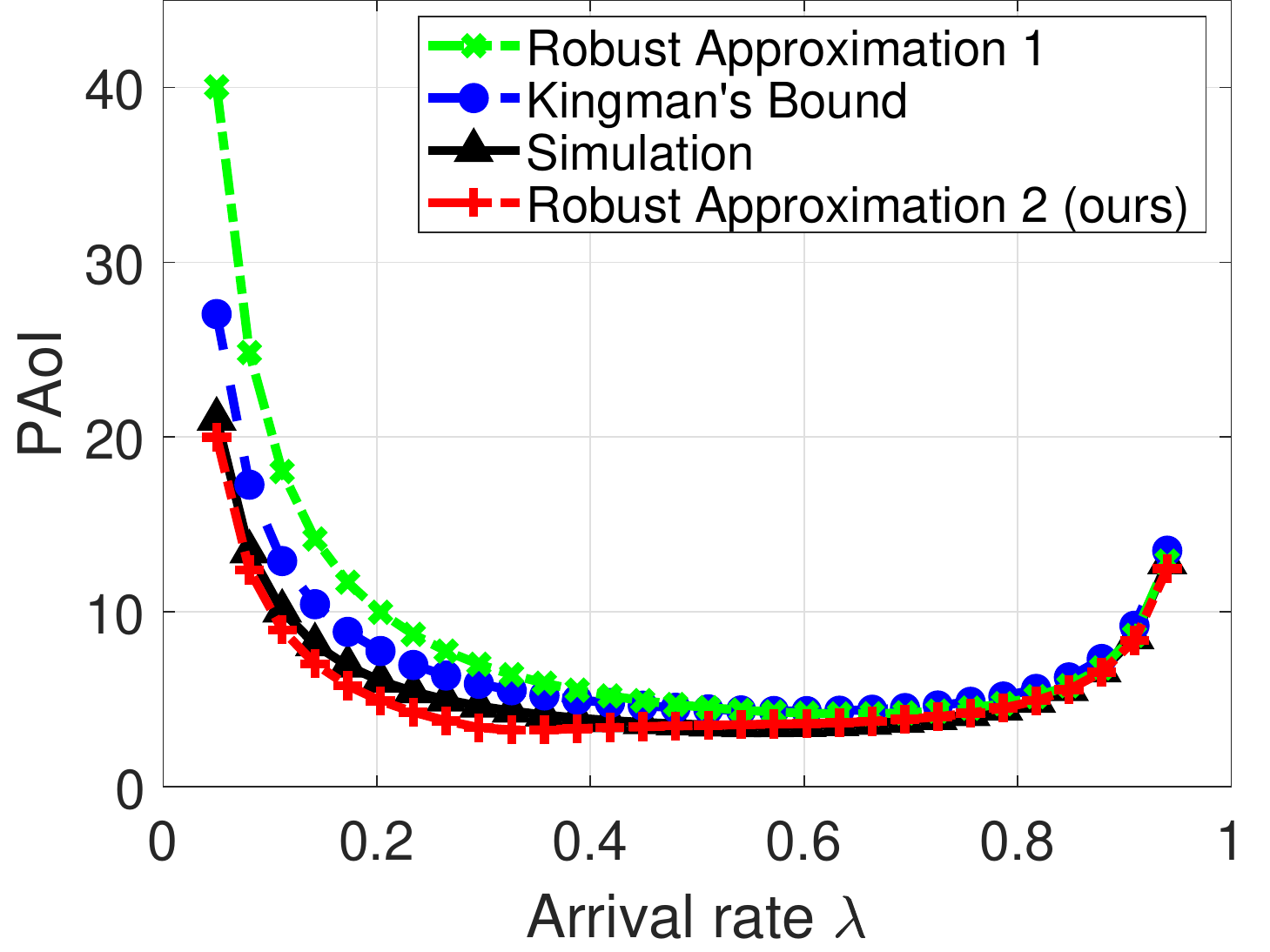}
	\caption{Normal distribution}
	\label{fig:single_halfnormal}
\end{subfigure}
\begin{subfigure}[b]{0.329\linewidth}
	\centering
	\includegraphics[width=1\textwidth]{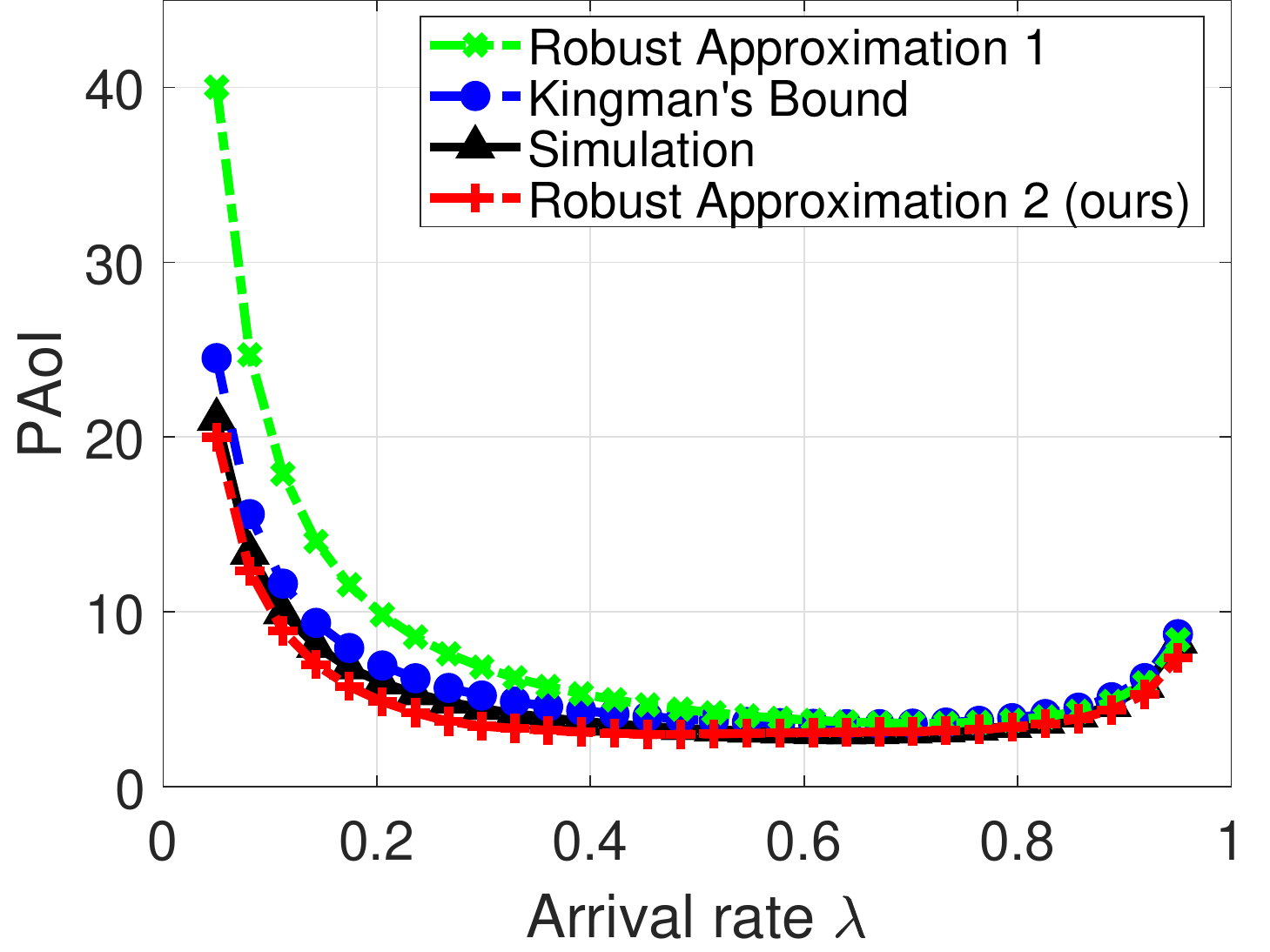}
	\caption{Uniform distribution}
	\label{fig:single_paretoT}
\end{subfigure}
\caption{Peak AoI under different distributions in the single-source setting.}
\label{fig:single}
\end{figure*}

\begin{figure*}[t]
\centering
\begin{subfigure}[b]{0.329\linewidth}
	\centering
	\includegraphics[width=1\textwidth]{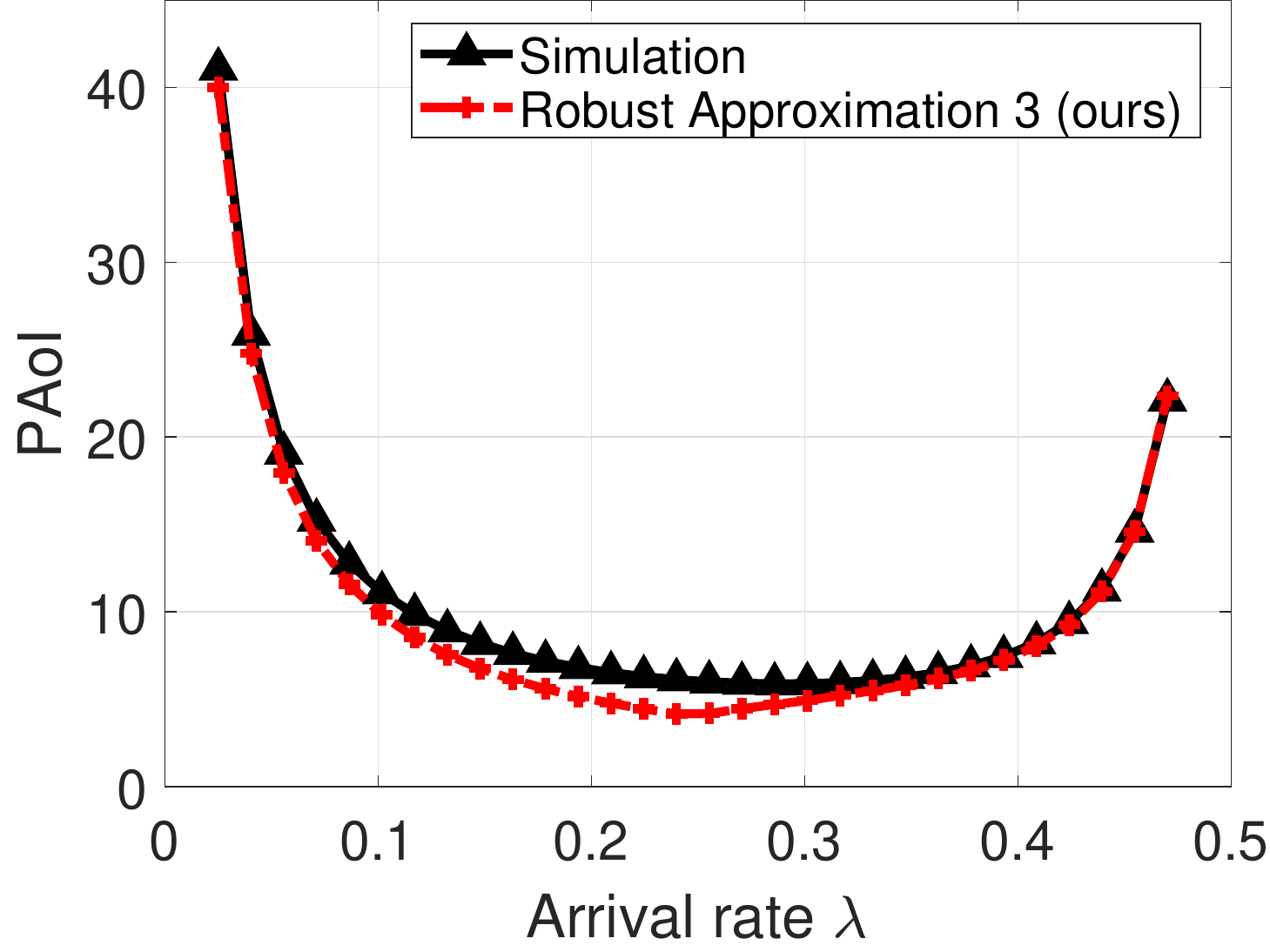}
	\caption{Exponential distribution}
	\label{fig:multi_poisson}
\end{subfigure}
\begin{subfigure}[b]{0.329\linewidth}
	\centering
	\includegraphics[width=1\textwidth]{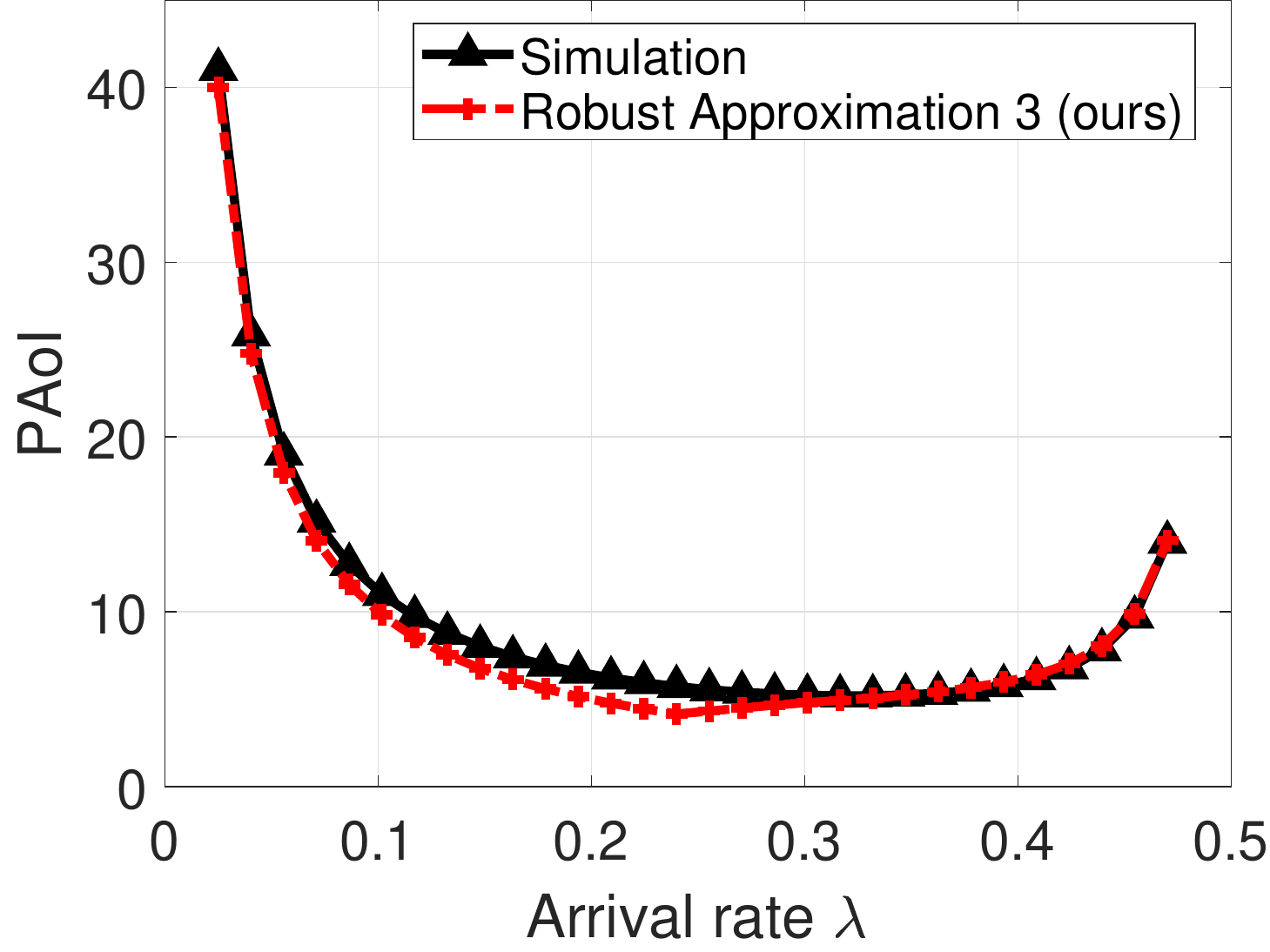}
	\caption{Normal distribution}
	\label{fig:multi_halfnormal}
\end{subfigure}
\begin{subfigure}[b]{0.329\linewidth}
	\centering
	\includegraphics[width=1\textwidth]{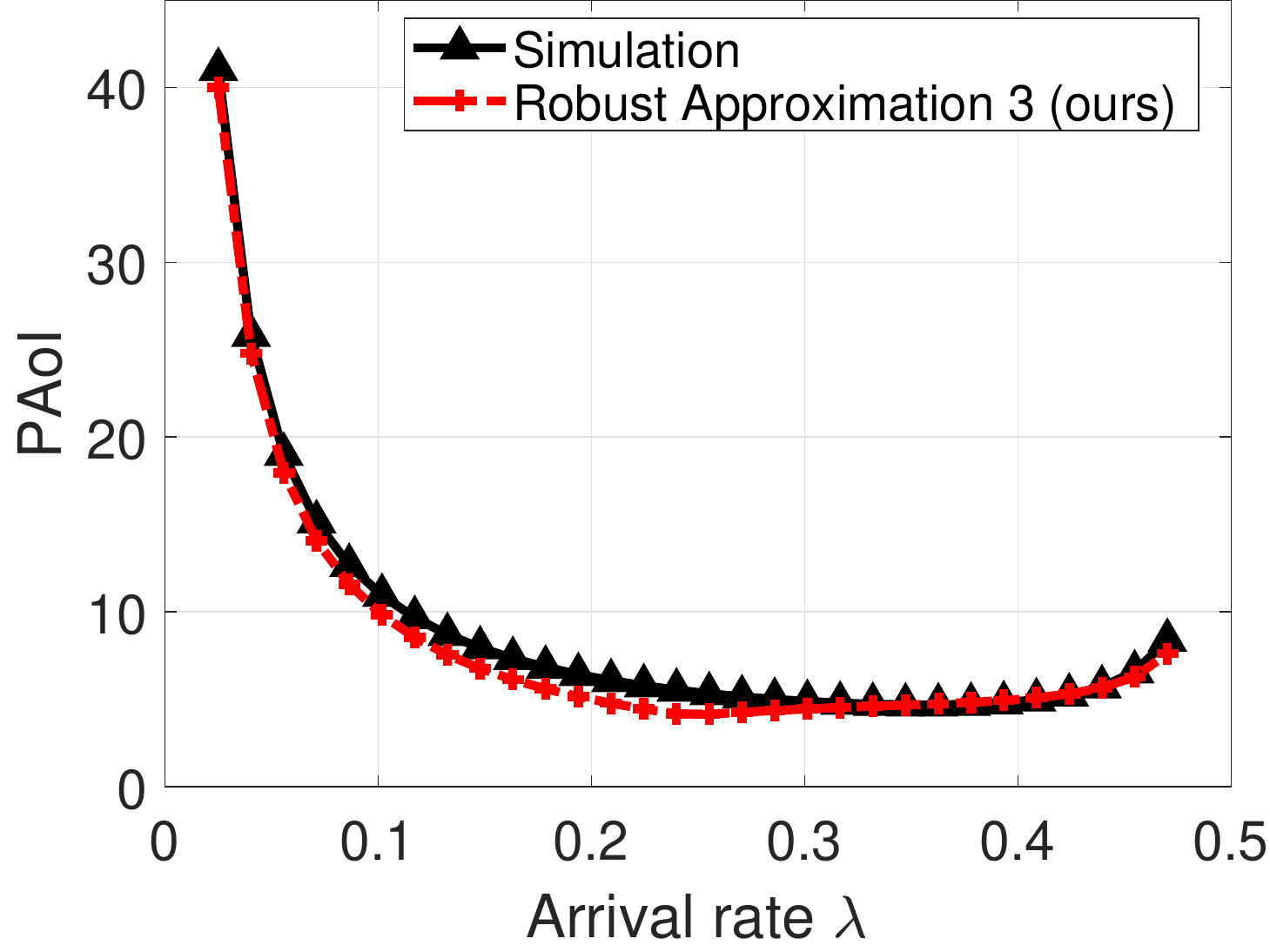}
	\caption{Uniform distribution}
	\label{fig:multi_paretoT}
\end{subfigure}
\caption{Peak AoI under different distributions in the two-source setting with symmetric arrivals.}
\label{fig:multi}
\end{figure*}

\subsection{Single-Source Systems}

We first consider the single-source setting, where the source generates updates with rate $\lambda$, and the service rate of the server is fixed with $\mu = 1$.
We consider three different distributions for the interarrival times and service times: i) exponential distribution, ii) normal distribution, and iii) uniform distribution.
For the normal distribution, we take the absolute value so that only positive interarrival time and service time are used.
We consider the normal distribution and uniform distribution in order to show that our bound does not rely on the properties of specific distributions (e.g., the memoryless property of the exponential distribution).

In Fig.~\ref{fig:single}, we plot the PAoI as the arrival rate $\lambda$ increases when the service rate $\mu = 1$. 
We name the approximation derived from \cite{robust} and Theorem~\ref{theo:singlebound} as \emph{Robust Approximation 1} and \emph{Robust Approximation 2}, respectively.
Note that the variability parameters used in our Robust Approximation 2 are from Table~\ref{table:adaption}.
Then, we compare the Kingman's bound,  Robust Approximation 1 and Robust Approximation 2 against the simulated PAoI. 
Though the Kingman's bound and  Robust Approximation 1 are originally defined for the expected system time, they can also be used as the PAoI bounds by simply adding the expected interarrival time $1/\lambda$.
In order to demonstrate the detailed performance of different bounds, the error percent (which is the sum of the difference between the simulated PAoI and the bound normalized by the simulated PAoI in each arrival rate, divided by the total number of arrival rates) of each bound are also given in Table~\ref{table:error_percent_single_source}. 
We run the simulations with a large number of arrivals ($n>10^4$) to ensure that the steady state is reached.
For each setting, we run the simulation for $50$ rounds and take the average. 

First, we can see that the PAoI is large under both light and high loads for all the considered distributions. This is because under the light load, the update arrival rate is low and the interarrival time between updates is large, which leads to a large PAoI; while under the high load, lots of updates wait to be served in the queue and the delay of those updates becomes large, which also leads to a high PAoI. This implies that we need to approximate PAoI well under both light and high loads. 
Second, we observe that Robust Approximation 2  we proposed can approximate the PAoI very well in all different settings.
From Fig.~\ref{fig:single_halfnormal} and Fig.~\ref{fig:single_paretoT}, we can see that the robust queueing approach is quite general and works well under normal and uniform distributions as well. 

Compared to Robust Approximation 1, Robust Approximation 2 is much closer to the simulation results when the traffic load is light.
This complies with our theoretical analysis, since Robust Approximation 2  cares about both light load and high load while Robust Approximation 1 is developed more towards approximating the system time under the high load.
For the Kingman's bound, it is known that the result is tight with light-tailed distribution in high load.
However, similar to Robust Approximation 1, the approximation of Kingman's bound is also loose in the light load.
Therefore, our proposed bound provides a competitive alternative to approximating the PAoI performance in information-update systems.

\subsection{Two-Source Systems}

In this subsection, we consider the information-updating system with two symmetric sources and one single server. Here the two sources generate updates with the same interarrival time distribution (also the same expected interarrival time 1/$\lambda$) and the same service time distribution (also the same expected service time 1/$\mu$), i.e., they are symmetric. 
Similar to the single-source case, we also consider three different distributions, where the interarrival times and service times are both exponential distribution, normal distribution and uniform distribution, respectively.
We name the approximation derived from Theorem~\ref{theo:multibound} as \textit{Robust Approximation 3}, which adopts the variability parameters in Table~\ref{table:adaption}.
With increasing arrival rate $\lambda$ and fixed service rate $\mu = 1$, the PAoI performances of one arbitrary source are shown in Fig.~\ref{fig:multi} and Table~\ref{table:error_percent_single_source}.

Note that the Kingman's bound and Robust Approximation~1 cannot be used to approximate the PAoI in the two-source setting since they are only derived for the single-source systems. Again, we observe that Robust Approximation~3  can approximate the average PAoI well under both light and high traffic loads.


\section{Conclusion}\label{sec:conclusion}
In this paper, we applied the robust queueing theory to analyzing the AoI performance in the communication systems.
By modeling the uncertainty in the stochastic arrival and service processes using  uncertainty sets, we provided a robust bound of the worse-case delay that can be used to approximate the expected PAoI in the single-source single-server systems.
Furthermore, we generalized our bound to the two-source single-server systems with symmetric arrivals. 
We showed that our bounds work well under both light and high loads and outperform prior bounds (e.g., the Kingman's bound), which could be quite inaccurate under the light load. 
Moreover, our results do not rely on the property of specific distributions and can thus be widely applied to more general settings.

There are several of interesting questions that are worth future work. For example, we only show that the robust queueing bound works well for the light-tailed systems in the simulations. However, the robust queueing theory is also capable of modeling the heavy-tailed behaviors. It would be interesting to investigate the effectiveness of the proposed bounds in the heavy-tailed systems. 
\textcolor{black}{In addition, we assume that  the service time and interarrival time have the same tail coefficient. It would be nice to generalize the results for the case where the service time and interarrival time have different tail coefficients. Finally, it is also worth studying the impact of the dependence between the arrival and service process \cite{whitt2018using}.
}


\bibliographystyle{IEEEtran}
\bibliography{reference}

\end{document}